\documentclass[11pt,a4paper]{article}

\usepackage{amsfonts,amssymb,amsthm,bm} %
\usepackage[longnamesfirst]{natbib}
\bibpunct{(}{)}{;}{a}{,}{,}
\usepackage[normalem]{ulem} %
\usepackage[tbtags]{amsmath} %
\usepackage{enumerate}
\usepackage{fullpage}
\usepackage{hyperref}
\usepackage{colortbl}

\usepackage{tikz}
\usepackage{neuralnetwork}

\definecolor{Green}{RGB}{0,200,0}
\definecolor{LightGreen}{RGB}{212,251,121}
\definecolor{Red}{RGB}{204,0,0}
\definecolor{LightRed}{RGB}{255,153,153}

\usepackage{graphicx,color}
\graphicspath{{./pics/}}
\definecolor{BrickRed}{rgb}{.625,.25,.25}

\definecolor{markergreen}{rgb}{0.6, 1.0, 0}
\definecolor{darkgreen}{rgb}{0, .5, 0}
\definecolor{darkred}{rgb}{.7,0,0}

\theoremstyle{plain}
\newtheorem{theorem}{Theorem}
\newtheorem{proposition}[theorem]{Proposition}
\theoremstyle{definition} 

\usepackage{mathrsfs}

\newcommand{\E}{{\mathbb{E}}}
\providecommand{\R}{{\mathbb{R}}}

\newcommand{\e}{\ensuremath{e}}

\newcommand{\dd}{{\rm d}}

\ifx\prop\undefined
\newtheorem{prop}{Proposition}[section]
\fi

\begin{document}
\title{\Large\bf Risk factor aggregation and stress testing}
\author{Natalie Packham\footnotemark} %
\maketitle %

\begin{abstract}
  Stress testing refers to the application of adverse financial or
  macroeconomic scenarios to a portfolio. For this purpose, financial
  or macroeconomic risk factors are linked with asset returns,
  typically via a factor model. We expand the range of risk factors by
  adapting dimension-reduction techniques from unsupervised learning,
  namely PCA and autoencoders.  This results in aggregated risk
  factors, encompassing a global factor, factors representing broad
  geographical regions, and factors specific to cyclical and defensive
  industries.  As the adapted PCA and autoencoders provide an
  interpretation of the latent factors, this methodology is also
  valuable in other areas where dimension-reduction and explainability
  are crucial.
\end{abstract}

\medskip

\footnotetext[1]{%
  \noindent Natalie Packham, Berlin School of Economics and Law,
  Badensche Str.\ 52, 10825 Berlin, Germany. Email:
  packham@hwr-berlin.de \medskip
  
  This research was partially supported by IFAF Berlin (Institut f\"ur
  angewandte Forschung Berlin e.V.) under the IFAF Explorativ scheme.
}

\section{Introduction}
\label{sec:introduction}

Stress testing refers to a set of methods and tools that assess the
impact of an adverse scenario on a financial portfolio. An adverse
scenario could, for example, be described as a downturn of
macroeconomic and financial risk factors.  Typically, a factor model
links the risk factors with asset returns, which in turn allows to
calculate the impact of the stress scenario on a portfolio.  Using
techniques from statistics and machine learning, we extend the
universe of risk factors by aggregating existing risk factors into
higher-level risk factors, such as a global risk factor, broad
geographic regions or cyclical and non-cyclical industries. The
methods developed also allow to evaluate the strength or weakness over
time of aggregated risk factors, such as the intensity of global risk,
which changes substantially over time.

The main underlying ideas employ techniques from unsupervised
learning, in particular Principal Component Analysis (PCA)
\citep{Jolliffe2002} and Autoencoders (AEs)
\citep{Goodfellow2016}. Both methods act on high-dimensional data,
with PCA a rotation of the data spaces' axes such that the resulting
representation of the data is orthogonal across dimensions while
successively maximising the variance in each dimension. This produces
a linear representation of the data through latent factors called
principal components (PCs), where discarding lower dimensions
optimally retains a maximum amount of variance. Conceptually, AEs can
be thought of as extending PCA to non-linear factors. In both cases,
the resulting factors are latent, and -- except in specialised
settings, such as interest rate term structures, see e.g.\
\citep{Litterman1991,Frye1997,Jamshidian1996,Loretan1997} -- have no
canonical interpretation. We refer to e.g.\
\citep{Laloux2000,Avellaneda2010} for a discussion on the
interpretability of PCs from equity returns.

In this paper, first we apply PCA to create a global stress
scenario. It turns out that at least two PCs must be retained from the
latent factor specification to address a global scenario. In order to
determine the number of relevant PCs and their ability to explain
certain risk factors, we use the Kaiser-Guttman-criterion and the
so-called participation ratio. Second, we adapt PCA and the AE by
first clustering the risk factors, and applying PCA / AE to each
category. Combining latent factors from each categories into on model
allows to give the resulting factors an interpretation.  It turns out
that the modified AE, which we call {\em clustered AE} calibrates
better to the risk factor data than the {\em clustered PCA}. This is
an indication that non-linearities are present in the data. Given the
temporal nature of the historical data lends itself to adding a long
short-term memory (LSTM) component to the clustered autoencoder. This
outperforms clustered PCA, but interestingly performs worse than a
simple clustered AE.

Given the aggregated risk factors, we demonstrate their use in stress
testing by applying stress scenarios to the global risk factor model,
European risk factor and the cyclical industries factor. All stress
scenarios are applied to equally-weighted portfolios of the DAX
constituents and the S\&P 500 constituents. It turns out that
clustered PCA and clustered AE methods produce similar stress test
results, with the clustered AE with LSTM producing stronger stress
impact and the simple clustered AE the least stress impact.

The fields of explainable Artificial Intelligence (XAI) and
explainable Machine Learning (XML) are vastly expanding, as missing
transparency and missing interpretability are one aspect holding back
the use of Machine Learning (ML) techniques in finance. More
specifically, XML attempts to relate predictions and forecasts of ML
models to the inputs. The ability to do so is becoming more important
with regulatory changes, such as the EU-AI Act \citep{EC2021}; see
also \citep{EBA2021a} for a discussion of ML models in the context of
financial risk models by the European Banking Authority (EBA). The
methods developed in this paper can be easily applied in the field of
XML to make black-box dimension reduction techniques explainable.

The paper is organised as follows: Section
\ref{sec:financ-stress-test} gives a brief overview of factor models
and the classical stress testing methodology. The data used in this
paper is explained in Section \ref{sec:data}. The different methods of
aggregating risk factors are developed in Section
\ref{sec:princ-comp-analys}. Section \ref{sec:stress-test-appl}
provides an example application of stress testing and Section
\ref{sec:conclusion} concludes.

All computations were done in Wolfram Mathematica.

\section{Financial stress testing}
\label{sec:financ-stress-test}

Stress testing refers to a diverse set of methods investigating the
value and risk of financial portfolios under adverse scenarios.  In a
classical setting, a factor model links financial asset returns with
observable risk factors, such as geographic regions and industries. A
stress scenario is then typically defined as either a hypothetical or
a historical scenario on the risk factors, with its impact evaluated
on a financial position. The simplest model linking asset returns and
risk factors is a linear regression model, e.g.\
\citep{Kupiec1998,Dowd2002,Jorion2007}. \citep{Bonti2006} use this
approach to illustrate stress testing of credit
portfolios. \citep{Packham2019,Packham2023} adapt the linear factor
model approach to link correlations with risk factors in order to
apply stress test scenarios on asset correlations. Further approaches
can be found in e.g.\ \citep{Alexander2008}.\footnote{Regulators
  require financial institutions to conduct stress tests, see e.g.\
  the European Capital Requirements Regulation (CRR), Article 290:
  ``[The stress testing programme for CCR (credit counterparty risk)]
  shall provide for at least monthly exposure stress testing of
  principal market risk factors such as interest rates, FX, equities,
  credit spreads, and commodity prices for all counterparties of the
  institution, in order to identify, and enable the institution when
  necessary to reduce outsized concentrations in specific directional
  risks.'' and ``It shall apply at least quarterly multifactor stress
  testing scenarios and assess material non-directional risks
  including yield curve exposure and basis risks. Multiple-factor
  stress tests shall, at a minimum, address the following scenarios in
  which the following occurs: (a) severe
  economic or market events have occurred; ...''. See\\
  \href{https://www.eba.europa.eu/regulation-and-policy/single-rulebook/interactive-single-rulebook/1880}{https://www.eba.europa.eu/regulation-and-policy/single-rulebook/interactive-single-rulebook/1880}.\\
  See also \citep{EBA2021} for details on regulatory requirements of
  stress testing.  }

Assuming $d\ll p$, a factor model expresses a vector of asset returns
$(r_1, \ldots, r_p)$ as a linear model of {\em risk factor returns\/}
$F_1, \ldots, F_d$ via
\begin{equation}
  \label{eq:5}
  r_i = \alpha_i + \beta_{i1} F_1 + \beta_{i2} F_2 + \cdots +
  \beta_{id} F_d + \varepsilon_i
  , \quad i=1, \ldots, p,
\end{equation}
where $
\beta_{i1}, \ldots, \beta_{id}
$ are {\em factor coefficients\/}, {\em factor weights} or {\em factor
  loadings},
$\alpha_i$ is a constant and $\varepsilon_i$ is the {\em idiosyncratic
  component} of the $i$-th asset return. It is common to assume that
the idiosyncratic components are uncorrelated. The factors
$F_1, \ldots, F_d$ are typically observable, such as index returns of
geographic regions and industries. 
Examples of factor models in credit risk management are Moody's KMV
\citep{Crosbie2002} and CreditMetrics \citep{Gupton1997}, see also
\citep{Bluhm2003,Bonti2006}.

The ``classical'' approach to stress testing separates risk factors
into a subset $\mathbf F_s$ of the risk factors, 
the {\em core\/} factors that are stressed and the complement
$\mathbf F_u$ of $\mathbf F_s$, the so-called {\em peripheral\/}
factors that are only indirectly affected by the stress
scenario. Assuming that the covariance matrix of risk factors remains
unaltered by the stress scenario and further assuming the risk factors
are normally distributed, the optimal estimator of $\mathbf F_u$
conditional on $\mathbf F_s$ is (e.g.\ Theorem \S 13.2 of
\cite{Shiryaev1996})
\begin{equation}
  \label{eq:6}
  \E(\mathbf F_u| \mathbf F_s) = \E(\mathbf F_u) + \Sigma_{us}
  \Sigma_{ss}^{-1} (\mathbf F_s-\E (\mathbf F_s)),
\end{equation}
where $\Sigma_{us}$ is the matrix containing the covariances of
$\mathbf F_u$ and $\mathbf F_s$ and where $\Sigma_{ss}$ is the
covariance matrix of $\mathbf F_s$.

\section{Data}
\label{sec:data}

\begin{table}[t]
  \centering
  \begin{tabular}{|c|c|c|}
    \cline{1-1}\cline{3-3}
    {\bf Geographical regions} && {\bf GICS industries}\\\cline{1-1}\cline{3-3}
    {\bf Europe} && {\bf Cyclical industries} \\\cline{1-1}\cline{3-3}
    Europe && 
              Materials \\
    France &&
              Industrials \\
    UK && 
          Consumer Discretionary\\
    Italy && 
             Financial\\
    Germany &&  
               IT\\\cline{1-1}
    {\bf Asia Pacific} &&
                          Real Estate \\\cline{1-1}\cline{3-3}
    Pacific && {\bf Defensive industries}\\\cline{3-3}
    Singapore &&
                 Energy\\
    Japan &&
             Consumer Staples\\
    Hong Kong &&
                 Health Care\\
    Australia && 
                 Communication Services\\\cline{1-1}
    {\bf North America} &&
                           Utilities \\\cline{1-1}\cline{3-3}
    United States \\
    Canada \\\cline{1-1}
    {\bf Emerging markets}  \\\cline{1-1}
    EM Latin America \\
    EM Europe + Middle East + Africa \\
    EM Asia \\\cline{1-1}
  \end{tabular}
  \caption{Stock market indices used as proxies for risk factors. All
    indices are MSCI stock indices. The indics are split into six
    categories, which serve as a basis for building aggregated
    indices. Data source: Refinitiv Eikon.}
  \label{tab:factors}
\end{table}

The primary data set is composed of MSCI stock indices representing
the observable risk factors. The indices consist of 16 geographic
regions and 11 industries according to the MSCI Global Industry
Classification Standard (GICS). Table \ref{tab:factors} shows the
indices as well as their (manual) classification into six categories,
which will be the basis for building aggregated risk factors. The data
set consists of daily equity data from Refinitiv Eikon in the period
from January 1999 until February 2023 (6214 observations). Figure
\ref{fig:scattermatrix} in the appendix shows a scatter matrix with
distribution properties of the data.

In addition to this data set, we will work with stock returns of the
DAX firms and S\&P 500 firms when devising stress tests. In order to
be able to include recent firms, the length of this data set is
limited to three years of daily data (750 observations).

\section{Aggregated risk factors}

To aggregate existing risk factors into higher-order risk factors, we
employ and modify PCA and autoencoders, commonly used
dimension-reduction techniques.

\subsection{Principal Component Analysis}
\label{sec:princ-comp-analys}

{\em Principal component analysis (PCA)\/} refers to a representation
of a random vector (or multivariate sample data points) as a linear
factor model with unobservable (latent) factors obtained from the
random vector (or the data) itself. The main idea underlying PCA is a
rotation of the coordinate axes in such a way that the factors
represent coordinate axes that are orthogonal, with the first factor
capturing the maximum variance of the data, the second factor
capturing the second most variance, and so on. If the data are
sufficiently correlated, then one can reduce the dimension whilst
retaining a high proportion of the variance. Mathematically, PCA
relates to the eigendecomposition of a covariance or correlation
matrix, with the eigenvectors the principal components (PCs) and the
eigenvalues expressing the variance captured by each PC. PCA goes back
to \citep{Pearson1901} and \citep{Hotelling1933}. For a detailed
treatment, see e.g.\ \citep{Jolliffe2002,James2013,Murphy2022}.

Let the $n\times d$ data $\mathbf X$ be standardised. The {\em first
  PC\/} (called loading vector in \citep{James2013})
$\phi_{11}, \ldots, \phi_{d1}$ is determined from finding so-called
{\em scores\/}
\begin{equation}
  \label{eq:4}
  z_{i1} = \phi_{11} x_{i1} + \phi_{21} x_{i2} + \cdots + \phi_{d1}
  x_{id},\quad i=1,\ldots, n,
\end{equation}
that have largest sample variance, subject to the constraint
$\sum_{j=1}^d \phi_{j1}^2=1$. In other words, the first PC vector
solves the optimisation problem
\begin{equation*}
  \max_{\phi_{11}, \ldots, \phi_{d1}} \Bigg\{ \frac{1}{n} \sum_{i=1}^n
  \underbrace{\left(\sum_{j=1}^d \phi_{j1} x_{ij}\right)^2}_{=z_{i1}^2}\Bigg\} \quad\text{
    subject to }\quad \sum_{j=1}^d \phi_{j1}^2 =1. 
\end{equation*}
The second PC vector $\phi_{12}, \ldots, \phi_{d2}$ is obtained in a
similar way by determining the scores $z_{i2}$, $i=1, \ldots, n$, that
have maximum variance out of all linear combinations {\em
  uncorrelated\/} with $z_{i1}$, $i=1,\ldots,n$.  Higher principal
component are determined likewise.

In practice, principal components are found via the eigendecomposition
of the correlation matrix of $\mathbf X$ (which is assumed to be
standardised). The $n\times d$ matrix $\mathbf Z$ of scores and
$d\times d$ matrix $\mathbf \Phi$ of PCs can be written in compact
form as
\begin{equation}
  \label{eq:2}
  \mathbf Z = \mathbf X\  \mathbf \Phi. 
\end{equation}
The scores can be viewed as factors, giving a factor model
\begin{equation}
  \label{eq:1}
  \mathbf{X} = \mathbf {Z}\, \mathbf \Phi'. 
\end{equation}
The last equation follows, as the eigenvectors $\mathbf \Phi$ 
satisfy $\mathbf \Phi^{-1}=\mathbf \Phi'$, due to the symmetry of the
correlation matrix.

If $\mathbf X$ is such that $n$ is the number of dimensions and $d$ is
the number of features, then the above conducts ``PCA on the
features'', which is of predominant interest for most applications,
such as dimension reduction. In this case, the columns of
$\mathbf \Phi$ are the eigenvectors of $\mathbf X'\, \mathbf X$ (a
$d\times d$ matrix) and the PC scores $\mathbf Z$ -- a $n\times d$
matrix or a dimension reduced variant with fewer columns -- serve as
the latent risk factors. In some papers that apply PCA to finance
data, e.g.\ \citep{Loretan1997}, the authors conduct ``PCA on the
observations'', i.e., on the eigenvectors of $\mathbf X\, \mathbf X'$
(an $n\times n$ matrix). The (elementary) Proposition below shows that
in this case the eigenvectors correspond to a scaled version of the PC
scores of $\mathbf X'\, \mathbf X$ and vice versa. In other words, one
can equivalently use the eigenvectors of $\mathbf X\, \mathbf X'$ are
latent factors.

\begin{proposition}
  Let $\mathbf L$ be the $d\times d$ diagonal matrix with the
  eigenvalues $l_1, \ldots, l_d$ on the diagonal and $\mathbf \Phi$ be
  the $d\times d$ matrix of eigenvectors of $\mathbf X'\, \mathbf X$,
  and let $\tilde{\mathbf L}$, $\tilde{\mathbf \Phi}$ be the
  corresponding matrices of $\mathbf X\, \mathbf X'$,
  respectively. Then,
  \begin{enumerate}[(i)]
  \item $\tilde l_k= l_k$, for $k=1,\ldots,d$, and $\tilde l_k=0$, for
    $k>d$.
  \item The eigenvectors of $\mathbf X\, \mathbf X'$ are given as
    \begin{equation*}
      \tilde{\mathbf\Phi} = \mathbf X\, \mathbf\Phi\, \mathbf L^{-1/2} =
      \mathbf Z\, \mathbf L^{-1/2},
    \end{equation*}
    where $\mathbf Z$ is the $n\times d$ matrix of PC scores of
    $\mathbf X\, \mathbf X'$. 
  \item The eigenvectors of $\mathbf X'\, \mathbf X$ are given as
    \begin{equation*}
      \mathbf\Phi = \mathbf X'\, \tilde{\mathbf\Phi}\, \mathbf
      L^{-1/2} = 
      \tilde{\mathbf Z}\, \mathbf L^{-1/2},
    \end{equation*}
    where $\tilde{\mathbf Z}$ is the $d\times n$ matrix of PC scores
    of $\mathbf X'\, \mathbf X$.
  \end{enumerate}
\end{proposition}
\begin{proof}
  The singular value decompositions (SVDs) of $\mathbf X$ and
  $\mathbf X'$ are, cf.\ e.g.\ \citep{Jolliffe2002},
  \begin{align*}
    \mathbf X& =\tilde{\mathbf \Phi}\, \mathbf L^{1/2}\, \mathbf \Phi'\\
    \mathbf X' &= \mathbf \Phi\, \mathbf L^{1/2}\, \tilde{\mathbf \Phi}', 
  \end{align*}
  where $\tilde{\mathbf \Phi}'\, \tilde{\mathbf \Phi}$ and
  $\mathbf \Phi'\mathbf \Phi$ are identity matrices and $\mathbf L$
  are the eigenvalues.
  Using that the PC scores are
  $\mathbf Z=\mathbf X\, \mathbf \Phi = \tilde{\mathbf \Phi}\, \mathbf
  L^{1/2}\, \mathbf \Phi'\, \mathbf \Phi = \tilde{\mathbf \Phi}\,
  \mathbf L^{1/2}$ gives
  $\tilde{\mathbf\Phi} = \mathbf Z\, \mathbf L^{-1/2} = \mathbf X\,
  \mathbf \Phi\, \mathbf L^{-1/2}$ .
\end{proof}

\subsection{PC contributions}
\label{sec:interpretation-pcs}

One way of giving the PCs an interpretation is to evaluate the
correlations between the data and the scores (which are the projection
of the data to each PC). Still assuming that the $\mathbf X$ is
standardised, from \eqref{eq:1}, using that the PCs are uncorrelated,
have variances $\lambda_i$, $i=1, \ldots, d$, it follows that
\begin{equation*}
  \text{Corr}(x_{\cdot j}, z_{\cdot i})
  = \frac{\text{Cov}(x_{\cdot j}, z_{\cdot i})} {\sqrt{\lambda_i}} %
  = \frac{\E[\phi_{ji}
    z_{\cdot i} z_{\cdot i}]}{\sqrt{\lambda_i}}
  = \phi_{ji} \sqrt{\lambda_i},
\end{equation*}
$j=1,\ldots, d$.  This expresses that the correlation of data and
scores are just the PCs scaled by the PCs standard deviation (which
can be considered a measure of the ``importance'' of the
PC).\footnote{ Some strands of the literature, e.g.\ \citep{Kent1979},
  call the correlations between data and scores loadings, while other
  authors, e.g.\ \citep{James2013}, call $\phi_{ji}$,
  $i=1, \ldots, d$, $j=1, \ldots, d$, loadings. We stick with the
  latter.}

\begin{figure}[t]
  \centering
  \includegraphics[width=.7\textwidth]{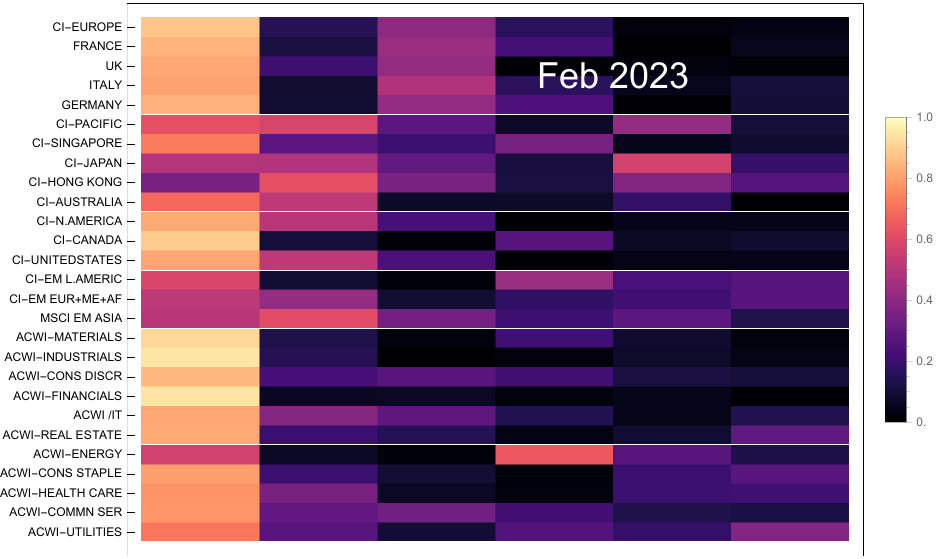}
  \caption{Heatmap of correlations of first six PCs for the risk
    factors. Each column shows the absolute correlations of the stock
    indicies with the scores of the respective PC. The stock indices
    comprise six categories: Europe, Asia Pacific, North America,
    Emerging Markets, cyclical indusries, defensive industries, which
    are separated by white lines. The correlations are calculated on a
    time window of 250 days at the end of February 2023.}
  \label{fig:correlations1}
\end{figure}

\begin{figure}[t]
  \centering \includegraphics[width=.45\textwidth]{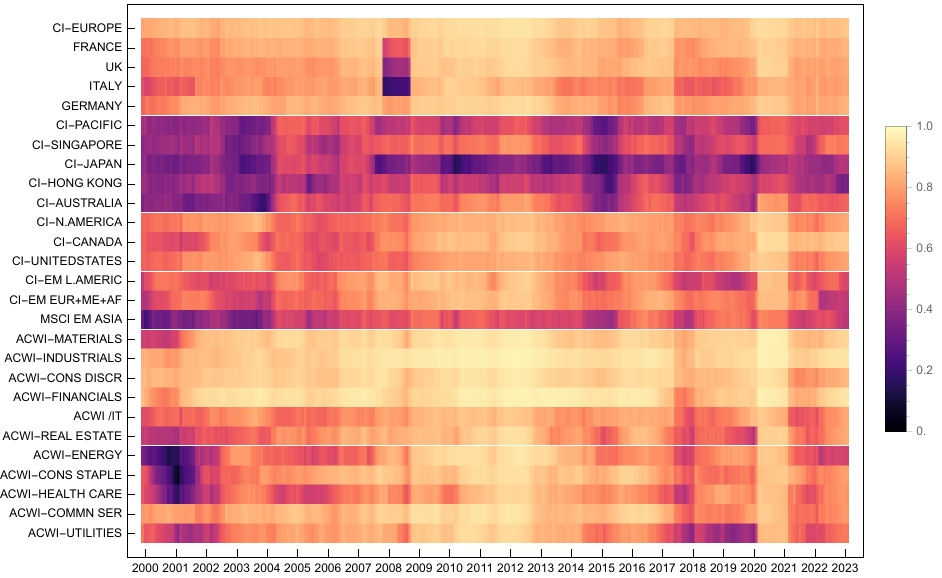}\ \
  \includegraphics[width=.45\textwidth]{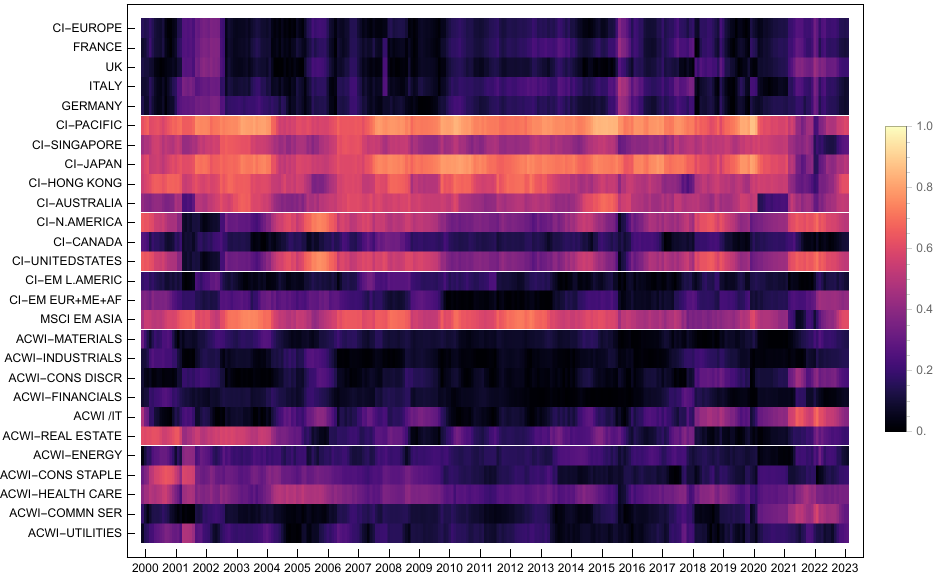}\medskip\\
  \includegraphics[width=.45\textwidth]{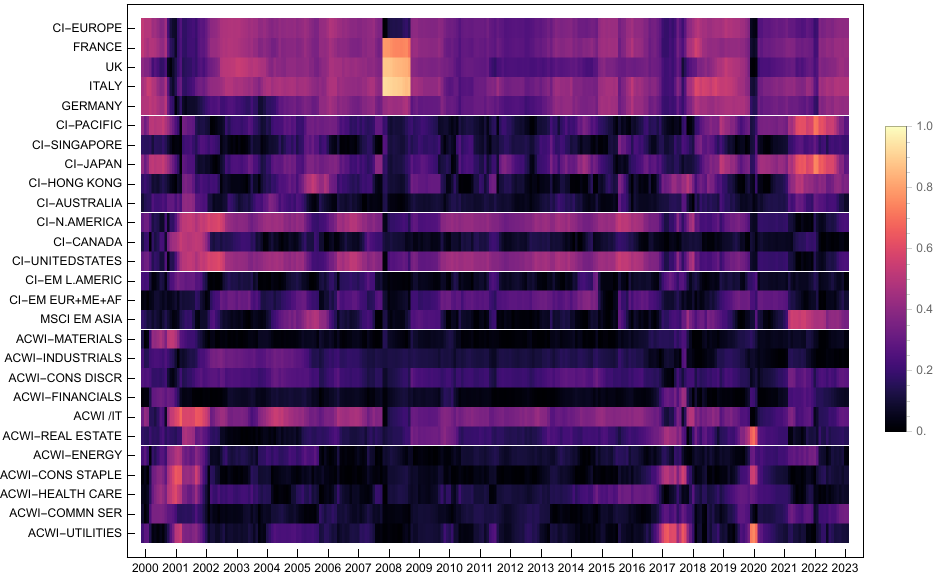}\ \
  \includegraphics[width=.45\textwidth]{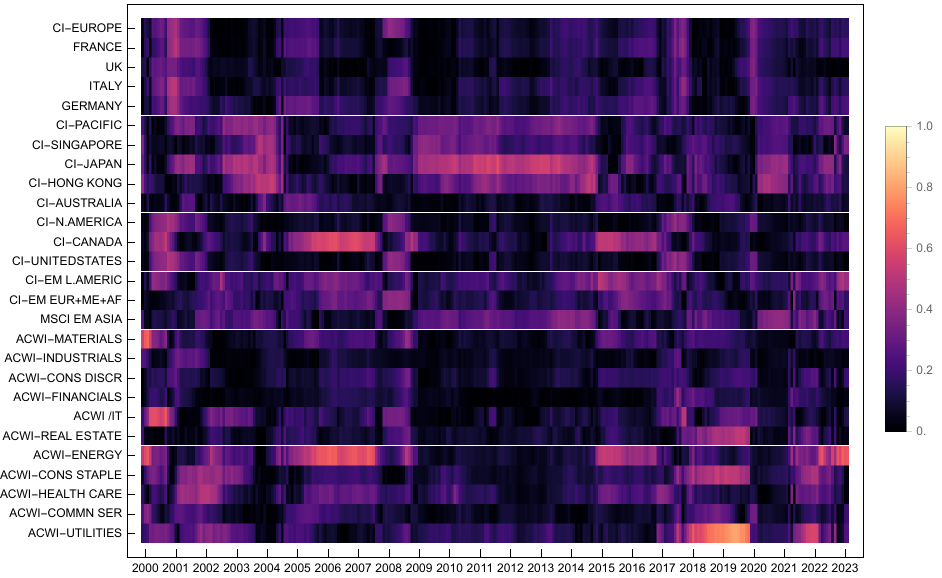}
  \caption{Correlations of the first four PCs through time, ordered
    top left, top right, bottom left, bottom right. Correlations are
    calculated on a rolling time window of 250 days.}
  \label{fig:correlations2}
\end{figure}

Figure \ref{fig:correlations1} shows the absolute correlations of the
first six PCs for each stock index in the set of risk
factors. Absolute correlations were used as we are primarily
interested in the strength of the correlations and as the sign of the
correlations is not uniquely specified. Figure \ref{fig:correlations2}
shows the absolute correlations of the first four PCs through time,
calculated at the end of each month on a time window of 250 days. As
expected, the first PC is the strongest factor. Quite surprisingly,
stock indices in the Asia-Pacific category have weaker correlations in
the first PC, but dominate the second PC. This effect is persistent
even if the data are changed to weekly data as opposed to daily data,
so it is unlikely to be a pure time zone effect. The third PC is
relatively strong for all European stock indices. The following
observations can be summarised: the first PC acts as a global risk
factor and the second PC acts as an Asia-Pacific risk
factor. Correlations change over time, so the strength of global risk
changes. For example, it is relatively weak in 1999-2000 (the
beginning of the data set) and in 2021-2023 (end of data set), while
it is strong in the period 2010-2012. We can also see that in several
months of the 2007-2008, European stock behaved differently from the
rest of the world.

Summarising, a global risk factor should incorporate the information
of the first two PCs.

\subsection{Relevance of PCs and PC interpretation}
\label{sec:relevance-pcs-pc}

\begin{figure}[t]
  \centering
  \includegraphics[width=.445\textwidth]{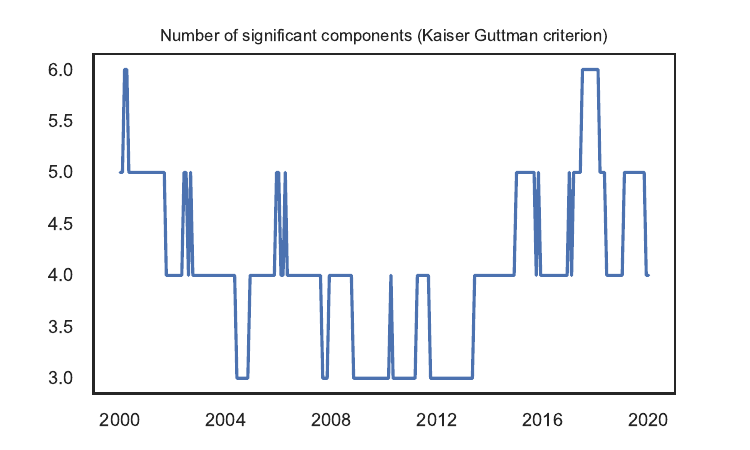}
  \includegraphics[width=.535\textwidth]{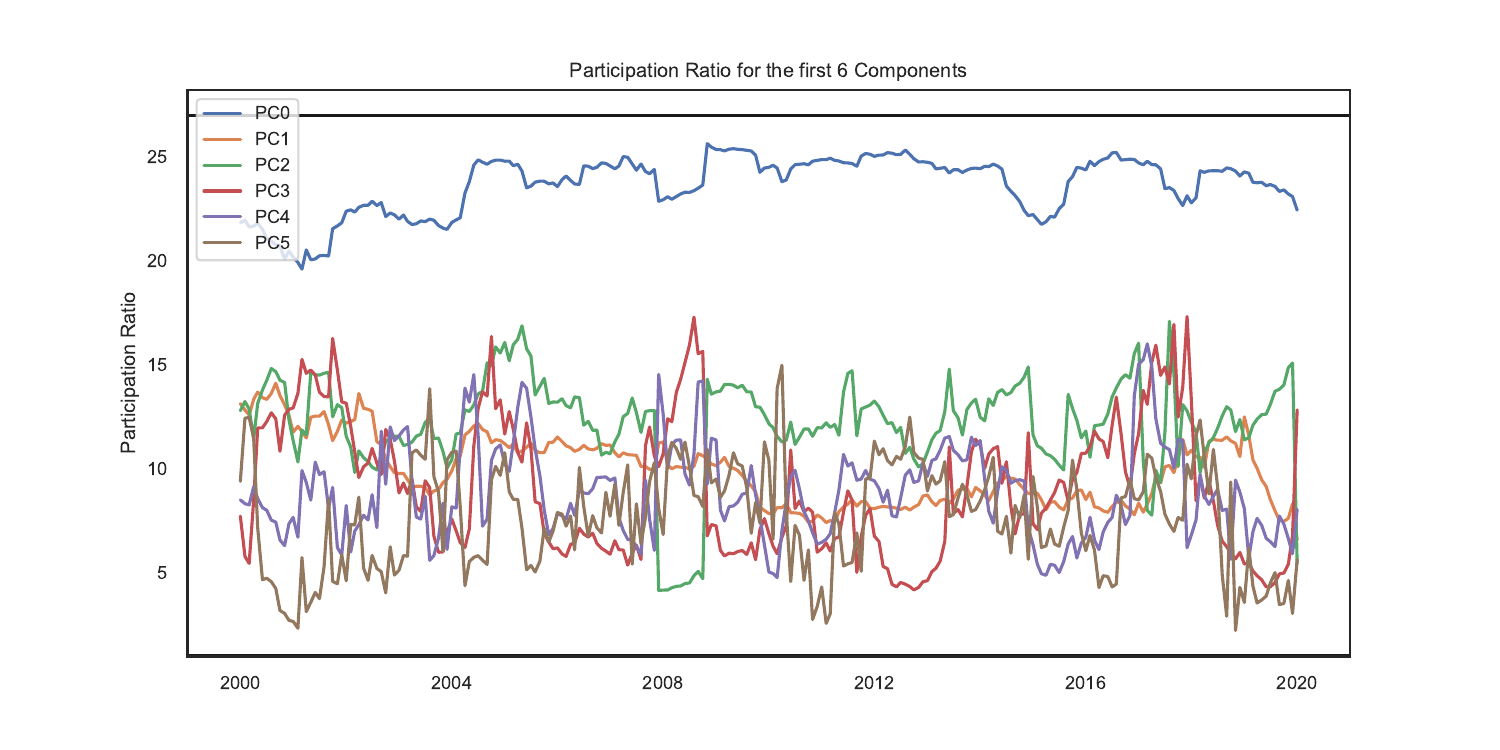}
  \caption{
    Left: Number of significant PCs according to the Kaiser-Guttman
    criterion.  Right: Participation Ratios (PR's) of first six
    PCs. The black line (at 27) corresponds to the number of risk
    factors in the data set.}
  \label{fig:pcrelevance}
\end{figure}

In the following, we formalise some of the observations of Figure
\ref{fig:correlations2}. This is relevant if the choice of the number
of relevant PCs or their contribution to risk factors is not as
clear-cut as in the setting above. More specifically, we review
methods to determine the number of PCs that are considered relevant,
and given the number of relevant PCs, we determine the risk factors
that contribute to these PCs. We also refer to \citep{Fenn2011} who
determine the number of significant PCs and their contributions on a
data set consisting of stock indices, bond indices, currencies and
commodities.

The {\em Kaiser-Guttman criterion\/} \citep{Guttman1954} specifies a
PC to be significant, it its (normalised) eigenvalue is greater than
$1/d$, where $d$ is the number of eigenvalues. The underlying idea is
that for uncorrelated data, each PC explains a variance fraction of
$1/d$ and each PC would be considered significant. The left-hand side
of Figure \ref{fig:pcrelevance} shows the number of relevant PCs
according to the Kaiser-Guttman criterion.

The {\em participation ratio (PR)\/} measures the contribution of the
risk factors to the PCs. The {\em inverse participation ratio (IPR)\/}
of the $k$-th PC is defined as \citep{Fenn2011,Guhr1998}
\begin{equation*}
  I_k = \sum_{j=1}^d (\phi_{jk})^4.
\end{equation*}
The IPR measures the number of assets participating in a PC: An
eigenvector with equal contributions, i.e., $\phi_{jk}=1/\sqrt{d}$,
has $I_k=1/d$, while the situation of an eigenvector with only one
contribution, i.e., $\phi_{jk}=1$, for one factor
$j\in \{1,\ldots, d\}$ and $\phi_{jk}=0$, for all other factors, has
$I_k=1$. As such, the IPR interpolates between $1/d$ and $1$. The PR
is defined as $1/I_k$, taking values between $1$ and $d$. A small
(large) PR therefore indicates that few (many) assets contribute to
the PC. The PR's of the first six PCs over time are shown in the
right-hand plot of Figure \ref{fig:pcrelevance}. Given that there are
27 risk factors, nearly all risk factors contribute to the the first
PC. The number of factors contributing to the second PC varies between
4 and 17 over time, while approximately half of the factors contribute
-- quite stably over time --to the third PC, with the exception of
several months 2007-2008, where it was already noticed that the third
PC captures European risk factors.

The PR can be employed to provide an classify the PCs according to the
six categories Europe, Asia Pacific, North America, Emerging market,
cyclical industries and defensive industries. For a given PC and its
PR, define the {\em PR group\/} as the group of size PR of indices
with the highest correlations. Each category is then associated with a
PC in exactly one of four ways:
\begin{center}
  \renewcommand{\arraystretch}{2}
  \begin{tabular}{cc}
    \hline
    \cellcolor{Green}
    \begin{minipage}{.35\linewidth}
      \vspace*{.5\baselineskip}
      {\bf Strong In}: \\
      {\em All} indices in a category are in the PR group.
      \vspace*{.5\baselineskip}
    \end{minipage}
    & \cellcolor{Red}
      \begin{minipage}{.35\linewidth}
        {\bf Strong Out}: \\
        {\em No} indices in a category are in the PR group.
      \end{minipage}\\
    \cellcolor{LightGreen}
    \begin{minipage}{.35\linewidth}
      \vspace*{.5\baselineskip}
      {\bf Weak In}: \\
      {\em More than half} of indices in a category are in the PR
      group.  \vspace*{.5\baselineskip}
    \end{minipage}
    & \cellcolor{LightRed}
      \begin{minipage}{.35\linewidth}
        {\bf Weak Out}: \\
        {\em Half or less} of indices in a category are in the PR
        group.  \vspace*{.5\baselineskip}
      \end{minipage}\\\hline
  \end{tabular}
\end{center}

\begin{figure}[t]
  \centering \includegraphics[width=.45\textwidth]{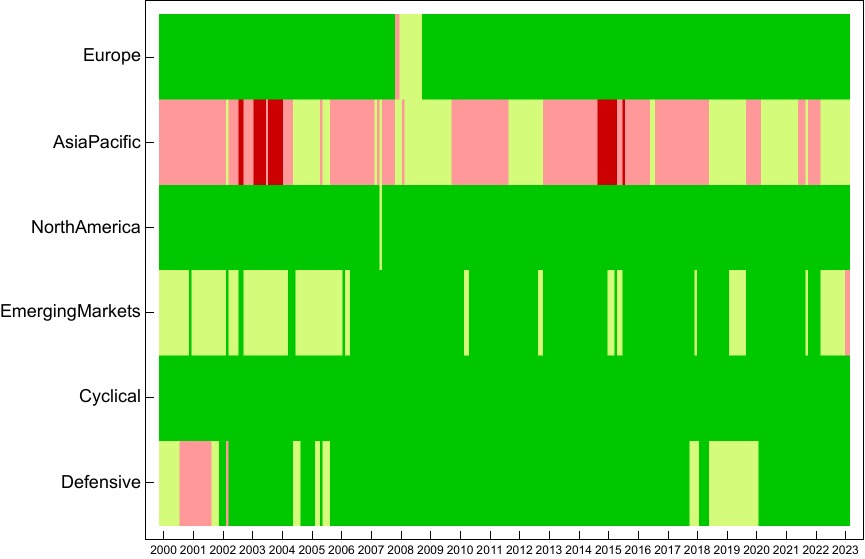}\ \
  \includegraphics[width=.45\textwidth]{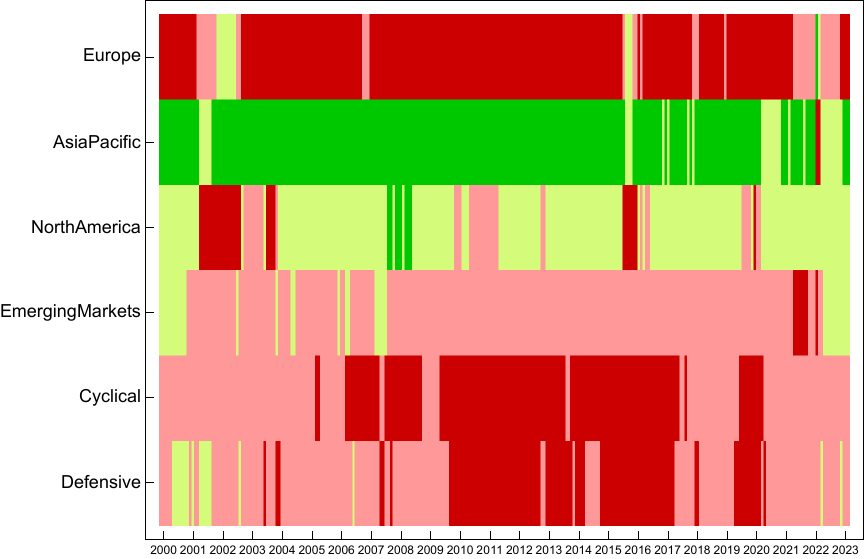}\medskip\\
  \includegraphics[width=.45\textwidth]{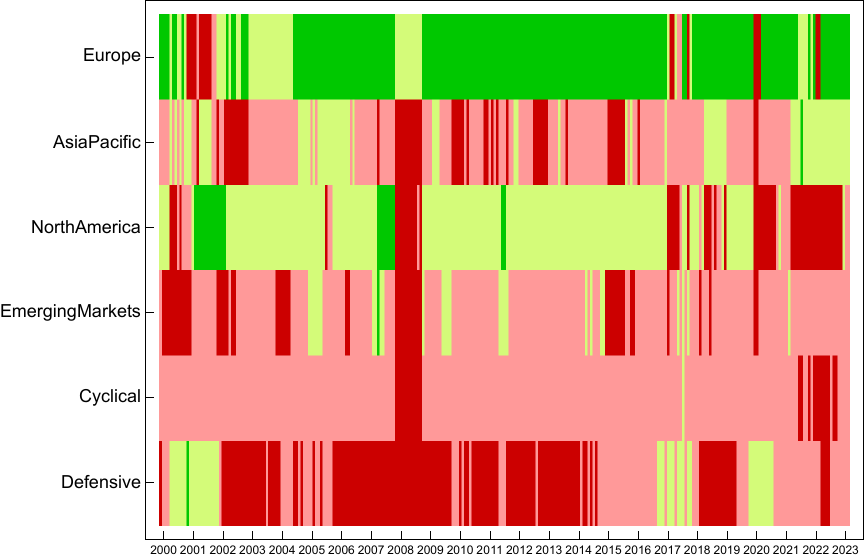}\ \
  \includegraphics[width=.45\textwidth]{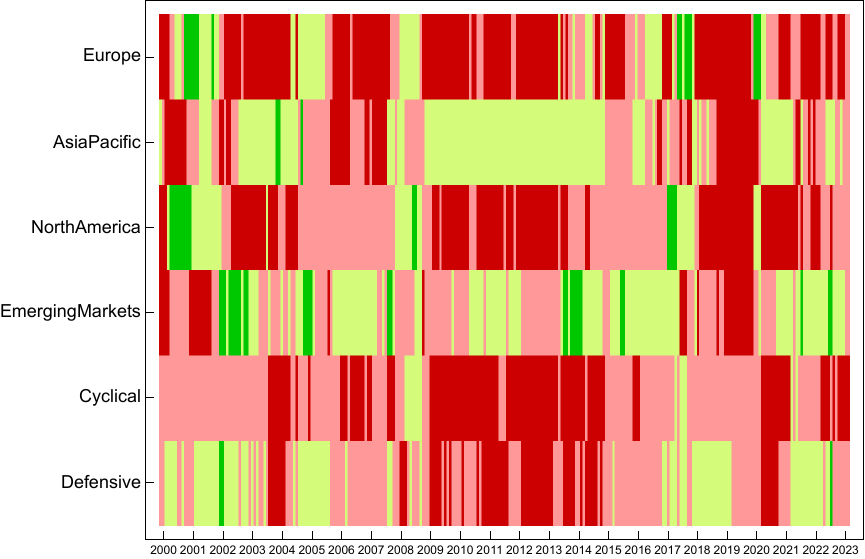}
  \caption{Assignments of stock index categories to PC explanation
    over time. Top left: first PC; top right: second PC; bottom left:
    third PC; bottom right: fourth PC.}
  \label{fig:pcainterpretation2}
\end{figure}

Figure \ref{fig:pcainterpretation2} shows the assignments of the
categories on a monthly basis over time. As can be seen, the first PC
(top left) acts as a global factors, mostly ex-Asia Pacific with few
exceptions. The second PC can be interpreted as the Asia Pacific
factor, capturing the correlation with North America. The thirds
factor explains Europe, capturing the joint behaviour of Europe and
North America. There is no clear pattern for the fourth factor.

\subsection{PCA within clusters}
\label{sec:pca-cluster}

The previous section revealed that plain PCA works well in building a
global risk factor, but does not succeed in establishing risk factors
for the individual risk categories (Europe, N.\ America, Asia-Pacific,
EM market, cyclical industries, defensive industries). This suggests
the following modification of PCA, which we call {\em clustered PCA\/}
in the following: conduct PCA on the data of each risk category
separately and combine the first PCs of each risk category into a
linear factor model. This approach is similar to the Hierarchical PCA
(HPCA) introduced by \citep{Avellaneda2020}.  More specifically, with
$K$ the number of clusters, denote by $z^{(k)}_{i1}$, $i=1,\ldots,n$,
the scores of the first PC of category $k$, $k=1,\ldots, K$, cf.\
\eqref{eq:4}.  The factor model is then established as (cf.\
\eqref{eq:5})
\begin{equation*}
  r_i = \alpha_i + \beta_{i1} z^{(1)}_{i1} + \cdots + \beta_{iK}
  z^{(K)}_{i1} + \varepsilon_i, \quad i=1, \ldots, n. 
\end{equation*}
Similar approaches have been devised by \citep{Enki2013}, see also
\citep{Mao2005,Masaeli2010,Chang2016}.

\begin{figure}[t]
  \centering \includegraphics[width=.9\textwidth]{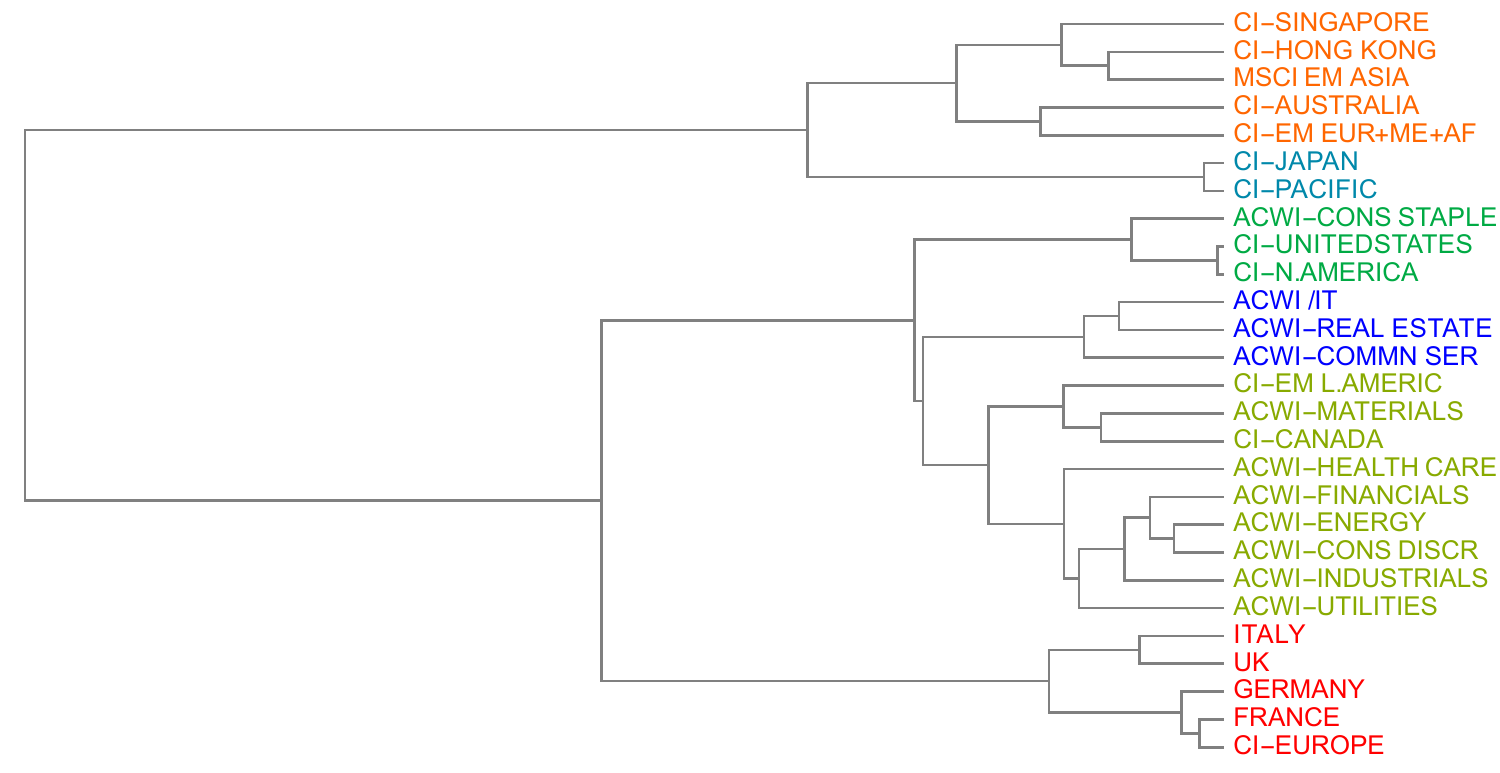}
  \caption{Dendrogram of the risk factors with Ward's minimum variance
    method as dissimilarity measure. The lengths of the connecting
    lines represent the degree of dissimilarity between
    clusters. Colour codes indicate the six clusters with highest
    dissimilarity.}
  \label{fig:clustering}
\end{figure}

If no categories are not readily available, then clustering methods
could be used to determine categories.  Figure \ref{fig:clustering}
shows the dendrogram obtained by hierarchical clustering using Ward's
minimum variance method as a dissimilarity measure
\citep{Ward1963}. This method minimises total within-cluster
variance. For details on Ward's method, we refer to
\citep{James2013,Johnson2007}. Truncating the dendrogram to give six
clusters yields clusters similar (but not equal) to the
categories. 

\begin{figure}[t]
  \centering
  \includegraphics[width=.32\textwidth]{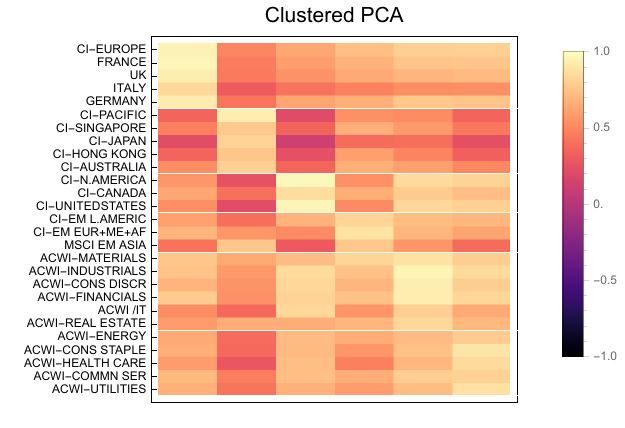}
  \includegraphics[width=.32\textwidth]{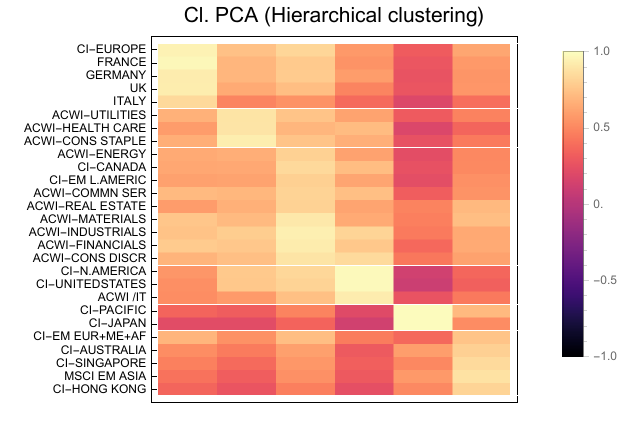}
  \includegraphics[width=.32\textwidth]{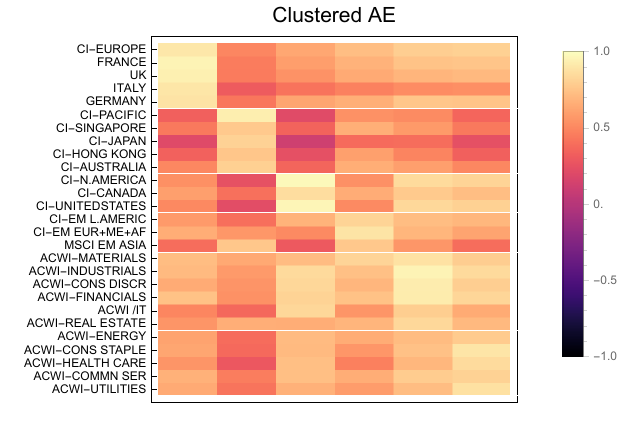}
  \caption{Left: Correlations of six first PCs from PCAs with the
    pre-specified categories. Middle: Correlations of clustered PCA
    with categories derived from hierarchical clustering; right:
    Correlations with scores obtained from clustered AE.}
  \label{fig:clustered_pca_1}
\end{figure}

Figure \ref{fig:clustered_pca_1} shows the correlations of the
original risk factors with the latent factors from clustered PCA.  The
left graph shows the correlations of the six first PCs obtained from
the six pre-specified categories, while the middle graph shows the
correlations obtained from applying the hierarchical clustering. The
high correlations associated with each category to which the first PC
belongs, are clearly visible. 

\subsection{Autoencoder}
\label{sec:autoencoder}

\begin{figure}[t]
  \centering
  \includegraphics[width=.5\textwidth]{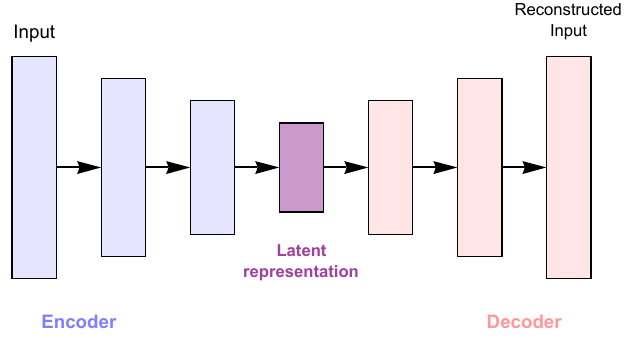}
  \caption[]{Schematic autoencoder representation. Source:
    Mathematica.\footnotemark{}}
  \label{fig:autoencoder_mathematica}
\end{figure}
\footnotetext{\url{https://reference.wolfram.com/language/ref/method/Autoencoder.html}}

\begin{figure}[t]
  \centering
  \includegraphics[width=.35\textwidth]{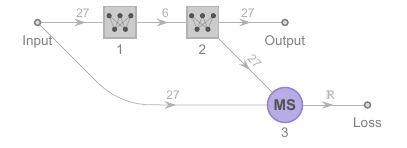}
  \includegraphics[width=.64\textwidth]{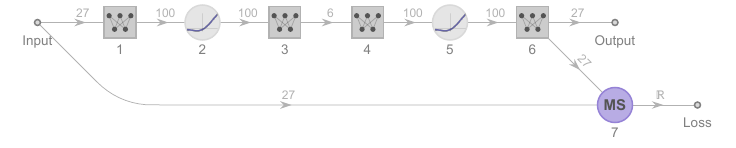}

  \includegraphics[width=.5\textwidth]{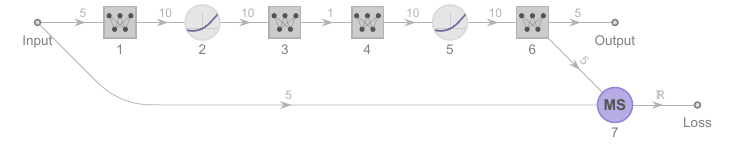}
  \includegraphics[width=.45\textwidth]{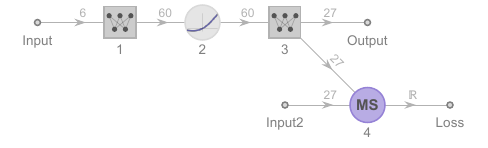}
  \caption{Top left: PCA implemented as an AE. The boxes denote linear
    layers with biases set to zero, while the circle labelled ``MS''
    calculates the mean square error (MSE) between its inputs. During
    training, the weights of the linear layers are adjusted to
    minimise the MSE output. The components of the network are
    numbered
    consecutively (indicated below each component). \\
    Top right: Optimal AE architecture obtained from training various
    architectures on the test data and validating on a separate
    validation set. The grey circles denote activation functions.\\
    Bottom left: Autoencoder of first category with bottleneck size
    one. The encoder part of AEs of each category (components 1, 2 and
    3) are taken as input to a joint decoder (bottom right). \\
    Bottom right: Joint decoder of clustered AE to reconstruct
    original data.}
  \label{fig:autoencoder_architectures}
\end{figure}

\begin{table}[t]
  \centering
  \begin{tabular}{llllll}
    Model & Specification & 
                                          L2-reg. & Batch size &  MSE \\\hline
    PCA & && & 0.1548\\
    AE & 100 / ``GELU'' & 
                                   $10^{-4}$ & 64& 0.1298\\
    AE with LSTM & {\small 100 / ``GELU'' / 27 / LSTM} & $1.25 \cdot 10^{-6}$ & 256 &
                                                                                      0.1411\\\hline
    Clustered PCA & & & & 0.1915\\
    Clustered AE & enc: 10 / ``Swish'' & - & 64 & 0.1741\\
          & dec: 60 / ``Swish''\\
    Clustered AE with LSTM & {\footnotesize enc: \# factors / ``SELU'' / LSTM} & $10^{-6}$ & 128 & 0.1782\\
          & {\footnotesize dec: LSTM / 27 / ``SELU'' / 27}
    
  \end{tabular}
  \caption{Autoencoder networks trained and the respective MSEs on the
    whole data set. The column ``Specification'' lists the layers of
    each network, where a number refers to a linear layer. ``GELU'',
    ``Swish'' and ``SELU'' refer to different activation functions. If
    not stated otherwise, then layers of the decoder are equal to
    layers of the encoder, in reverse order.  The ADAM optimiser was
    used throughout.}
  \label{tab:MSEs}
\end{table}

PCA and its variants capture linear relationships between risk factor
returns and the latent factors. To capture non-linear relationships,
we consider the {\em autoencoder (AE)}, a simple feedforward neural
network used to learn an efficient representation of data, typically
for the purpose of dimensionality reduction or feature learning. The
main idea is as follows: An AE consists of an encoder, which
compresses the data and a decoder, which reconstructs it. The final
layer of the encoder, called the {\em bottleneck\/} or {\em latent
  representation\/}, which connects to the decoder, has a smaller
dimension than the data. In this way, the AE learns a
dimension-reduced representation of the data. More specifically, an AE
consists of an encoder function $\mathbf h=f(\mathbf x)$ and a decoder
that produces a reconstruction $g(\mathbf h)$. The AE thus produces
$g(f(\mathbf x))$, and by restricting the code $\mathbf h$ to a lower
dimension than $\mathbf x$, the AE will not be able to copy
$\mathbf x$ exactly, but reproduce a copy that resembles the training
data.

As the AE is a feedforward neural network, it is composed of layers,
each of which applies a so-called activation function to a linear
function of its input and passes this as an output to the next
layer. Given $\mathbf x$, the output of the $i$-th layer is
\begin{equation*}
  \mathbf h^{(i)} = g^{(i)} \left(({\mathbf W}^{(i)})' \mathbf x +
    \mathbf b^{(i)}\right),
\end{equation*}
where $g^{(i)}$ is an activation function, ${\mathbf W}^{(i)}$ is a
matrix of weights and $\mathbf b^{(i)}$ is a vector of biases. Typical
activation functions are the sigmoid function, tanh, rectified linear
unit (ReLU), amongst many others, see e.g.\ \citep{Goodfellow2016}.

Figure \ref{fig:autoencoder_mathematica} shows a schematic
representation of an autoencoder.  Autoencoders go back to
\citep{Bourlard1988,Kramer1991}, see also \citep{Hinton2006}, which
presents a technique for effectively training deep AEs; see also the
book by \citep{Goodfellow2016} for a concise reference.

A shallow linear network architecture without activation functions
captures only the linear relationship between its input and the
bottleneck, and therefore yields the same dimension-reduced
representation of the data as PCA, see the left-hand-side of Figure
\ref{fig:autoencoder_architectures}. One can therefore think of
autoencoders as a generalisation of PCA. However, it should be taken
into account that the standard AE does not produce orthogonal latent
factors and therefore the code, i.e., the output of the bottleneck
layer, is not, in general, equal to the PC scores.

Because the data consists of return time series, the flexibility of
the AE architecture lends itself to more sophisticated layers that
take into consideration temporal dependencies. Introducing long
short-term memory (LSTM) layers, which are recurrent and as such
suitable for sequential data, is therefore a natural extension. An
LSTM cell is composed of four units: the input gate, the output gate,
the forget gate and the self-recurrent neutron. The LSTM goes back to
\citep{Hochreiter1997} and the LSTM-AE is described for example by
\citep{Sagheer2019}.

\subsection{Clustered AE}

Just as with clustered PCA, we modify the AE so it produces a code for
each category and join them into one decoder. The first step -- one AE
for each category -- is similar to the step of applying PCA separately
to each category, while the second step -- a unified decoder --
corresponds to the linear model reproducing the full original data.
An example of the separate AE is shown in the bottom left of Figure
\ref{fig:autoencoder_architectures}. The bottom right shows the joint
decoder. The right graph of Figure \ref{fig:clustered_pca_1} shows the
correlations obtained from the scores of each categorie's AE with the
original risk factors.

\subsection{Calibration results}
\label{sec:calibration-results}

Table \ref{tab:MSEs} shows the specification and mean square errors
(MSEs) of the PCA, AE, AE with LSTM models as well as their clustered
counterparts. To determine the optimal network architecture, the AE
networks were first all trained on a test data set comprising 80\% of
the data and validated on the remaining 20\%. This involved different
numbers of layers, different sizes of layers, different activation
functions as well as different parameters for the optimiser, such as
the parameter for $L2$-regularisation. Choosing the architecture with
the smallest validation MSE, the final network was trained on the
whole data set. The final activation functions chosen are:
\begin{itemize}
\item ``GELU'': Gaussian error linear unit, with activation function
  $g(x) =\displaystyle\frac{1}{2} x \left(1+
    \text{erf}\left(\frac{x}{\sqrt{2}}\right)\right)$, where
  $\text{erf}(x)=\displaystyle \frac{2}{\sqrt{\pi}} \int_0^x
  \e^{-t^2}\, \dd t$ is the error function;
\item ``SELU'': scaled exponential linear unit, with activation
  function
  \begin{equation*}
    g(x)=\displaystyle
    \begin{cases}
      1.0507 x,&\text{ if } x\geq 0,\\
      1.7581 (\e^x-1)&\text{ if } x<0.
    \end{cases};
  \end{equation*}
\item ``Swish'': $g(x) = \displaystyle\frac{x}{1+\e^{-x}}$.
\end{itemize}

Interestingly, the optimal AE and clustered AE consist of just one
layer each in the encoder and in the decoder. Deep AEs, with several
hidden layers, exhibit a low training MSE and a high validation MSE,
an indication of overfitting. The added flexibility when compared to
PCA lies in the greater size of the linear layer compared to the input
as well as non-linear activation functions. Although adding an LSTM
layer provides a better fit than PCA, this does not outperform the
simpler AE. Summarising, a simple AE architecture with non-linear
activation provides the best fit, indicating the presence of
non-linear relationships in the data.

\section{Stress testing application}
\label{sec:stress-test-appl}

As an application of the aggregated risk factors, we conduct several
stress scenarios on equally-weighted portfolios of DAX firms as well
as S\&P 500 firms.  In order to be able to include firms that were
recently added to the stock indices, only 750 observations of daily
returns (approx.\ three years) were used to calibrate the factor
models \eqref{eq:5}. Given that the AE models capture non-linearities,
one would of course prefer to train decoder type networks (cf.\ bottom
right of Figure \ref{fig:autoencoder_architectures}), but many stocks
of interest do not have a sufficiently long data history and training
on 750 observations turns out to be very unstable.\footnote{ One could
  attempt to replace the missing data with synthetic data, which is a
  popular technique in machine learning, when the training data set is
  too small; see e.g.\
  \citep{Dogariu2022,Ni2021,Freeborough2022,Park2022} } Overall, 27
DAX firms and 496 S\&P 500 firms are included in the portfolios. The
following three scenarios are considered:
\begin{itemize}
\item Global stress scenario: The first two PCs are chosen to capture
  global risk. Scenarios where linear combinations of the factors are
  2 standard deviations from the mean are considered. This is equal to
  considering all joint scenarios that correspond to a Mahalanobis
  distance of 2 standard deviations from the mean (see e.g.\
  \citep{McNeil2015} for the Mahalanobis distance). Formally, for each
  portfolio, the worst outcome of these scenarios is chosen:
  \begin{equation*}
    \min_{\mathbf s=(s_1,s_2)'\in \R^2} \frac{1}{p} \sum_{i=1}^p
    r_i(s_1,s_2)=\frac{1}{p}\sum_{i=1}^p \alpha_i +
    \beta_{i1} F_1(s_1) + \beta_{i2} F_2(s_2), 
  \end{equation*}
  subject to
  $\displaystyle\sqrt{(\mathbf s-\mathbf\mu_F)'\, \Sigma_F^{-1}\,
    (\mathbf s-\mathbf\mu_F)}\leq 2$, which constrains the square root
  of the Mahalanobis distance. For both portfolios, the constraint
  turns out to be binding.
\item European stress scenario: the European risk factor in the
  clustered variants of PCA, AE and AE with LSTM is bumped by PC /
  clustered AE -2 standard deviations.
\item Cyclical industries stress scenario: the risk factor
  representing cyclical industries is bumped by -2 standard
  deviations.
\end{itemize}

\begin{figure}[t]
  \centering
  \includegraphics[width=.475\textwidth]{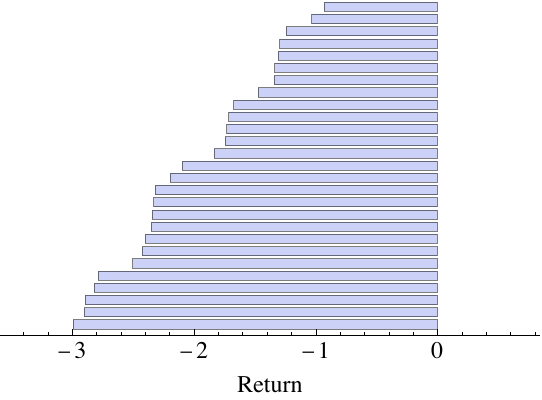}
  \includegraphics[width=.475\textwidth]{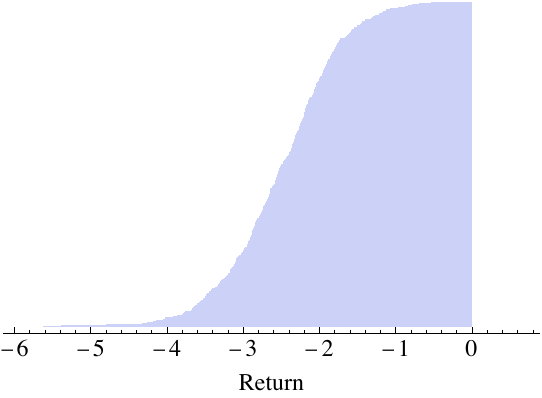}
  \caption{Impact of a two standard deviation downturn of global risk
    factor (first PC) on German DAX constituents (left) and S\&P 500
    constituents (first and second PC) (right).}
  \label{fig:globalpca}
\end{figure}

\begin{table}[t]
  \begin{center}
    \begin{tabular}{llll}
      \multicolumn{4}{c}{\bf Global stress scenario}\\\hline
      Index & First PC & Second PC & Impact\\\hline
      DAX& -2 s.d. &  0  & -2\%\\
      S\&P 500 & -1.7701 s.d. & -0.9309 s.d. & -2.2273\%
    \end{tabular}
  \end{center}
  \begin{center}
    \begin{tabular}{lccc}
      \multicolumn{4}{c}{\bf European stress scenario}\\\hline
      Index & Cl.\ PCA & Cl.\ AE & Cl.\ AE w/ LSTM \\\hline
      DAX & -2.5788\% & -1.99\% & -2.8267\% \\
      S\&P 500 & -1.5173\% & -1.3649\% & -1.9579\%
    \end{tabular}
  \end{center}
  \begin{center}
    \begin{tabular}{lccc}
      \multicolumn{4}{c}{\bf Cyclical industries stress scenario}\\\hline
      Index & Cl.\ PCA & Cl.\ AE & Cl.\ AE w/ LSTM \\\hline
      DAX & -2.0389\% & -1.7608\% & -2.0521\%\\
      S\&P 500 & -2.2277\% & -2.0008\% & -2.1934\%
    \end{tabular}
  \end{center}
  \caption{Mean impact on equally-weighted portfolio of different
    stress scenarios. The global stress scenarios are selected as the
    worst-impact scenarios 2 standard deviations from the mean. For
    the European and the cyclical industries stress scenarios, the
    respective risk factor was bumped by -2 standard deviations.}
  \label{tab:stressscenarios}
\end{table}

Figure \ref{fig:globalpca} shows the impact of the global stress
scenario on the individual stocks ordered from greatest impact to
lowest impact from bottom to top. The impact of all stocks is
negative. Table \ref{tab:stressscenarios} shows the impact on
equally-weighted portfolios. The worst scenario for the DAX
corresponds to a shift in the first PC, but no impact on the second
PC. For the S\&P 500 both PCs are shifted, with a stronger shift on
the first PC. Approx. 4\% of all realisations of the first PC are more
extreme than -2 standard deviations, so observing a daily return
smaller than the impact of 2\% should be expected to occur on ten days
per year.

\begin{figure}[t]
  \centering
  \includegraphics[width=.475\textwidth]{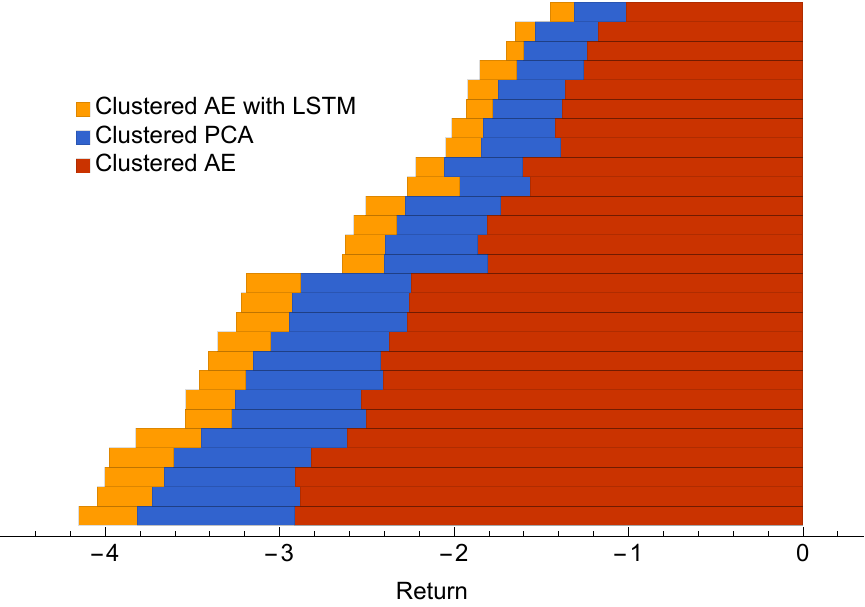}
  \includegraphics[width=.475\textwidth]{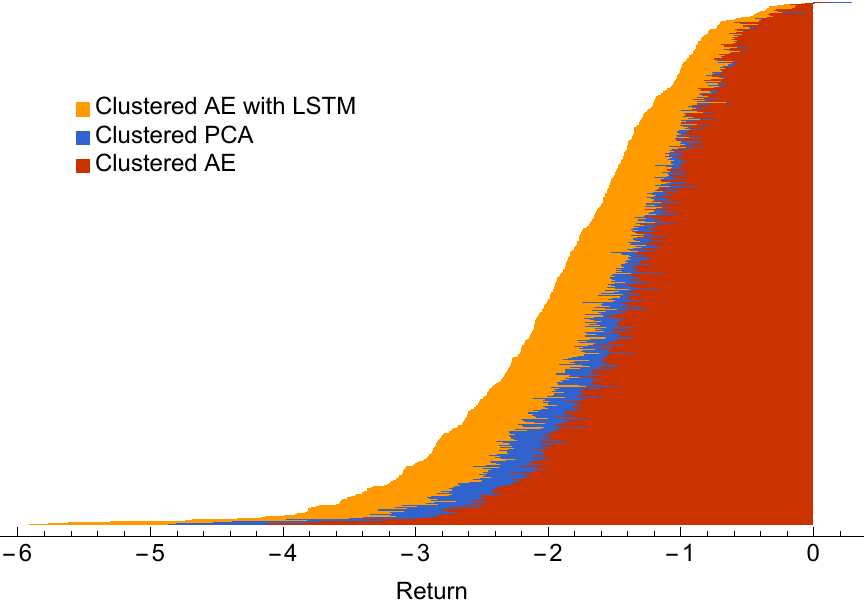}

  \includegraphics[width=.475\textwidth]{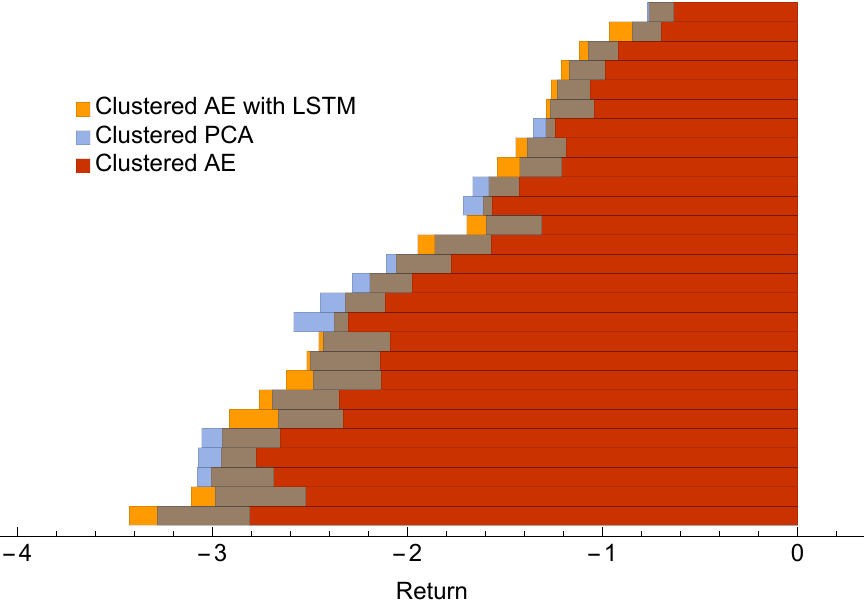}
  \includegraphics[width=.475\textwidth]{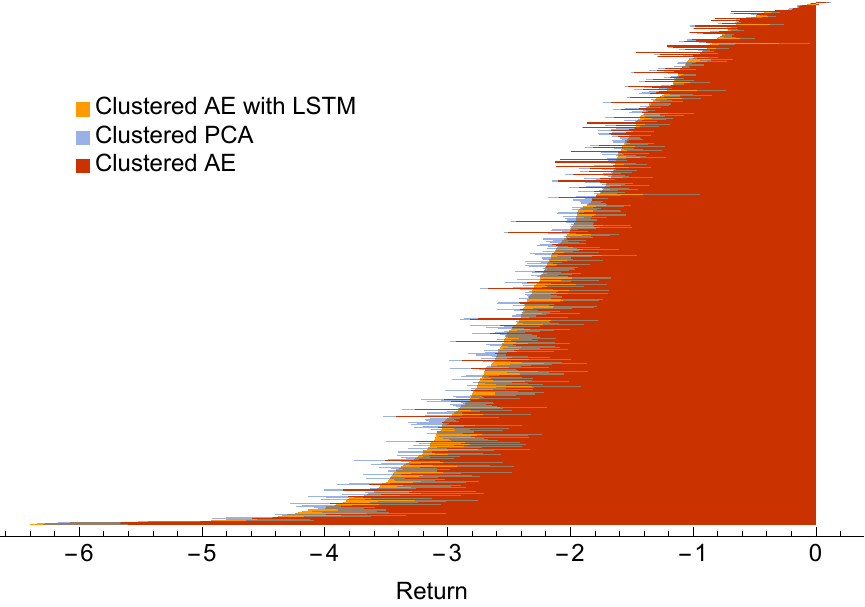}
  \caption{Impact of European stress scenario (top) and cyclical
    industries scenario (bottom) on DAX stocks (left) and S\&P 500
    stocks (right).}
  \label{fig:stressscenarios}
\end{figure}

For the other two stress scenarios, the European, resp.\ cyclical
industries, factor is considered the core risk factor and the
remaining five factors are the peripheral factors. In the case of
clustered PCA, the impact on the peripheral factors was calculated by
Equation \eqref{eq:6}, while for the clustered AEs the impact was
predicted from training a decoder-type network with the architecture
from the original clustered AE with the size of the linear layer
adjusted to the smaller sizes of the input and output layers.

The top of Figure \ref{fig:stressscenarios} shows the outcome of the
European stress scenario on the individual stocks. Approx.\ 6\% of the
realisations of the European risk factor are more extreme than the
adverse scenario of -2 standard deviations from the mean. This
expresses that an outcome at least as extreme as the observed scenario
should be expected to occur on 15 trading days per year. Depending on
the model used, the impact on an equally-weighted portfolio lies
between -2.83\% and -2\% for the DAX portfolio and in the range of
-1.96\% to -1.36\% for the S\&P 500 portfolio (see Table
\ref{tab:stressscenarios}). The clustered AE with LSTM produces the
most extreme impact, while the AE produces the least impact with the
clustered PC in-between.

Finally, for the cyclical industries factor, the bottom of Figure
\ref{fig:stressscenarios} indicates less variation in the impact of
clustered AE with LSTM and clustered PCA, with the clustered AE
producing slightly weaker impact. A scenario of -2 standard deviations
occured only twice in the 750-day period under consideration (with the
standard deviations determined from the longer history of observations
going back to 1999), which corresponds to 0.3\% of observations. This
suggests that a higher number of extreme outcomes of the cyclical
industries occured prior to the 750-day period. The impact is slightly
weaker on the DAX portfolio, ranging from -2.05\% to 1.76\%, than on
the S\&P 500 portfolio, ranging from -2.22\% to -2\%.

\section{Further applications and conclusion}
\label{sec:conclusion}

The aggregated risk factors derived in Section
\ref{sec:princ-comp-analys} have further applications, essentially in
all areas of finance where factor models are employed. For example,
estimating the $p\times p$ covariance matrix $\Sigma$ of asset returns
directly is noisy (i.e., subject to high confidence intervals), and
can be stabilised through factor models by the approximation
\begin{equation*}
  \Sigma \approx B\Omega B',
\end{equation*}
where $B$ is the $p\times d$ matrix of factor coefficients, $\Omega$
is the $d\times d$ covariance matrix of the risk factors and the
variances of the residuals are ignored. Reducing the number of factors
through aggregated factors may give an even more robust estimate.

The clustered AE can also be applied in ML and AI applications beyond
the finance space. Consider a prediction problem with a high number of
features. Dimension reduction methods, such as the AE, are commonly
used to increase the efficiency of training, but come at the expense
of losing interpretability of the driving factors. Combining
dimension-reduction techniques with clustering can be employed to
create aggregated features that are explainable.

To conclude, we employ and modify methods from unsupervised learning
to create aggregated risk factors from observable risk factors, such
as geographic regions and industries. More specifically, a global risk
factor is produced using PCA and AE, while modified versions, called
clustered PCA and clustered AE allow for constructing several
aggregated risk factors relating to larger geographic areas (e.g.\ a
European factor) and categories of industries, such as cyclical and
defensive industries. In all instances, AEs improve the representation
of the original risk factors through the aggregated factors over PCA,
which indicates the presence of non-linear effects. Incorporating an
LSTM component, however, to capture temporal dependencies, did not
outperform the simpler AE versions.

The aggregated risk factors can be used to build high-level stress
scenarios. This is demonstrated on portfolios consisting of DAX and
S\&P 500 components by considering a global stress scenario, a
European stress scenario and a stress scenario on cyclical industries.

\appendix

\section{Scatter plots of risk factors}
\label{sec:scatter-plot-risk}

\begin{figure}[t]
  \centering
  \includegraphics[width=\textwidth]{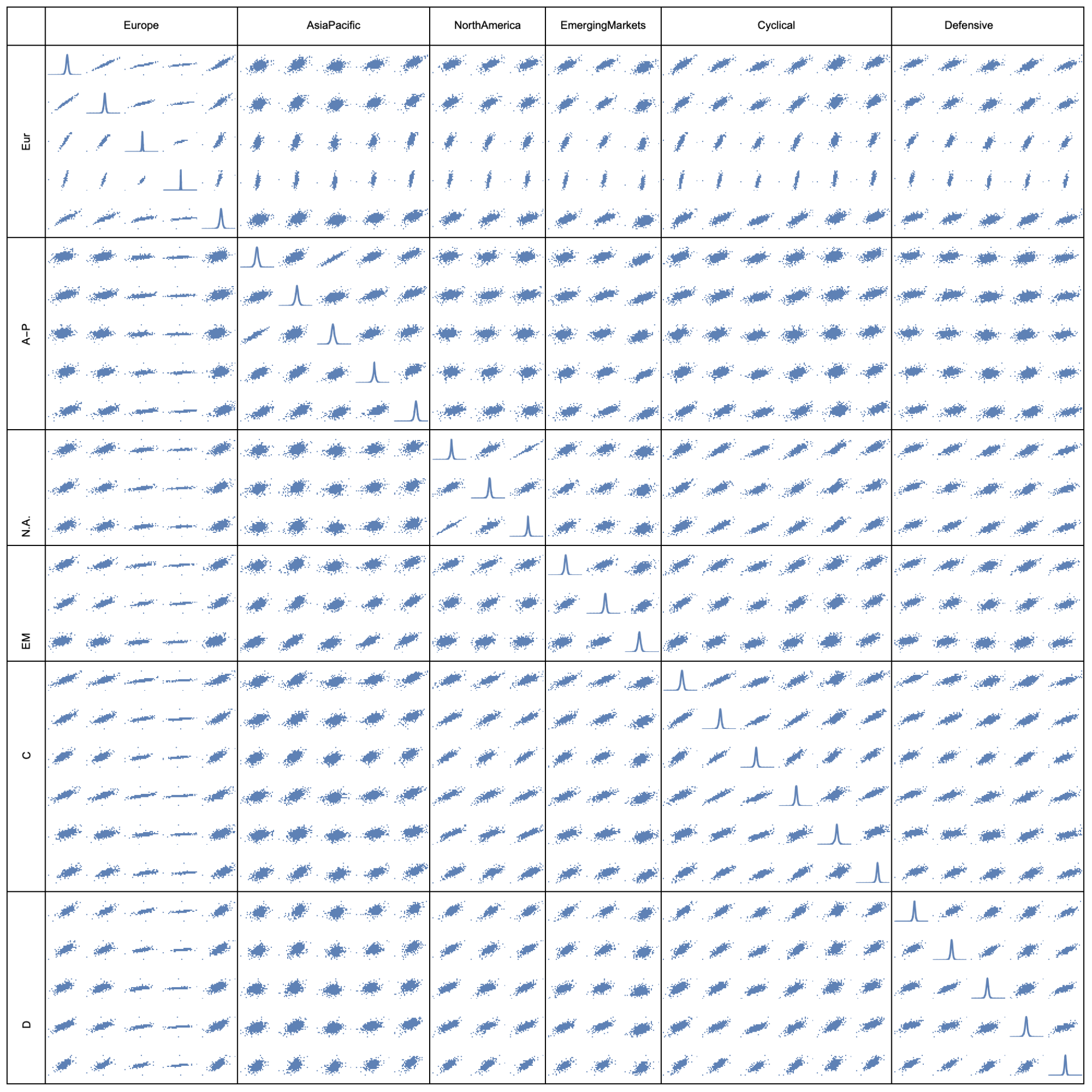}
  \caption{Scatter matrix of the 27 risk factors. Smooth kernel
    densities of each time series are shown on the diagonal, while the
  off-diagonal elements contain scatter plots of the respective time
  series. }
  \label{fig:scattermatrix}
\end{figure}

\bibliographystyle{abbrvnamed} %
\bibliography{finance} %

\end{document}